\title{%
	On the complexity of solving linear congruences and computing nullspaces modulo a constant
}
\author{%
	Niel de Beaudrap\thanks{
		This work was supported by the EC project QCS.}
	\\[1ex]
		\normalsize	DAMTP, Centre for Mathematical Sciences, University of Cambridge,	\\[-0.5ex]
		\normalsize	Wilberforce Road, Cambridge CB3 0WA, UK	
}
\date{%
	\normalsize
	5 August, 2013
}

\pdfoutput=1

\documentclass[draft,a4paper,10pt]{article}
\usepackage[utf8x]{inputenc}
\usepackage{amsmath,amssymb,amsthm,bm,bbm,mathtools}
\usepackage{scalefnt}
\usepackage{hyperref}
\usepackage{enumitem}
\usepackage[left=35mm,right=35mm]{geometry}

\hypersetup{
	pdftitle={On the complexity of solving linear congruences and computing nullspaces modulo a constant},
	pdfauthor={Niel de Beaudrap},
	colorlinks,%
	citecolor=black,%
	filecolor=black,%
	linkcolor=black,%
	urlcolor=black
}

\theoremstyle{definition}
\newtheorem{definition}{Definition}

\newtheorem*{problems}{Problems}

\theoremstyle{plain}
\newtheorem{theorem}{Theorem}
\newtheorem{lemma}[theorem]{Lemma}
\newtheorem{prop}[theorem]{Proposition}
\newtheorem{corollary}[theorem]{Corollary}

\newcommand\inter\cap

\renewcommand\subset\subseteq

\renewcommand\ge\geqslant
\renewcommand\le\leqslant

\renewcommand\vec\mathbf
\newcommand\unit[1][e]{\vec{\hat #1}}

\makeatletter

\newcommand\ox\otimes
\newcommand\x\times

\makeatletter
\newcommand\eg{\emph{e.g.}}
\newcommand\ie{\emph{i.e.}}
\newcommand\etc{\@ifnextchar.{\emph{etc}}{\emph{etc.}}}
\newcommand\etal{\@ifnextchar.{\emph{et al}}{\emph{et al.}}}
\makeatother

\DeclareMathOperator\dom{dom}
\DeclareMathOperator\Null{null}
\DeclareMathOperator\poly{poly}

\newcommand\N{\mathbb N}
\newcommand\Z{\mathbb Z}

\newcommand\F{\mathbb F}

\let\union\cup
\let\inter\cap

\newcommand\num[2][]{\ensuremath{\textbf{\upshape \#}{#2}_{#1}}}
\newcommand\FUL[1][]{\ensuremath{\mathsf{FUL}_{#1}}}
\newcommand\UL[1][]{\ensuremath{\mathsf{UL}}}

\renewcommand\P{\ensuremath{\mathsf P}}

\newcommand\co{\ensuremath{\mathsf{co}}}

\newcommand\Mod[1][]{\ensuremath{\mathsf{Mod}_{#1}}}
\newcommand\Log{\ensuremath{\mathsf{L}}}
\newcommand\NC{\ensuremath{\mathsf{NC}}}
\newcommand\AC{\ensuremath{\mathsf{AC}}}
\newcommand\Ceq{\ensuremath{\mathsf{C}_{\textbf=}}}
\newcommand\Gap{\ensuremath{\mathsf{Gap}}}

\newcommand\Poly{\ensuremath{\mathsf{poly}}}
\newcommand\Fun{\ensuremath{\mathsf{F}}}
\newcommand\Fundot{\ensuremath{\mathsf{F}{\cdot\:\!}}}

\newcommand\prob[1]{\mbox{\mdseries\upshape{\scalefont{0.95}{#1}}}}
\newcommand\bits[1][]{\prob{bits$(#1)$}}
\newcommand\LCON[1][]{\prob{LCON$_{#1}$}}
\newcommand\LCONNULL[1][]{\prob{LCONNULL$_{#1}$}}
\newcommand\LCONX[1][]{\prob{LCONX$_{#1}$}}

\newcommand\trans{^{\mathsf T}}
\newcommand\sur[1]{^{(#1)}}

\makeatletter
\renewcommand\paragraph{\@startsection{paragraph}{4}{\z@}%
                                    {3.25ex \@plus1ex \@minus.2ex}%
                                    {-1em}%
                                    {\normalfont\normalsize\itshape}}
\makeatother

\begin{document}
\maketitle

\begin{abstract}
	\noindent
		We consider the problems of determining the feasibility of a linear congruence, producing a solution to a linear congruence, and finding a spanning set for the nullspace of an integer matrix, where each problem is considered modulo an arbitrary constant $k \ge 2$.
		These problems are known to be complete for the logspace modular counting classes $\Mod[k]\Log = \co\Mod[k]\Log$ in special case that $k$ is prime~\cite{BDHM92}.
		By considering variants of standard logspace function classes --- related to \num\Log\ and functions computable by \UL\ machines, but which only characterize the number of accepting paths modulo $k$ --- we show that these problems of linear algebra are also complete for \co\Mod[k]\Log\ for any constant $k \ge 2$.

		Our results are obtained by defining a class of functions \FUL[k] which are low for \Mod[k]\Log\ and \co\Mod[k]\Log\ for $k \ge 2$, using ideas similar to those used in the case of $k$ prime in Ref.~\cite{BDHM92} to show closure of \Mod[k]\Log\ under $\NC^1$ reductions (including \Mod[k]\Log\ oracle reductions).
		In addition to the results above, we briefly consider the relationship of the class \FUL[k] for arbitrary moduli $k$ to the class \Fundot\co\Mod[k]\Log\ of functions whose output symbols are verifiable by \co\Mod[k]\Log\ algorithms; and consider what consequences such a comparison may have for oracle closure results of the form $\Mod[k]\Log^{\!\!\;\Mod[k]\Log} = \Mod[k]\Log$ for composite $k$.

		\medskip\noindent
		Keywords: Modular arithmetic, linear congruence, logarithmic workspace
\end{abstract}


\section{Introduction}

Solving a system of linear equations, or determining that it has no solution, is the definitive elementary problem of linear algebra over any ring.
This problem is the practical motivator of the notions of matrix products, inverses, and determinants, among other concepts; and relates to other computational problems of abelian groups, such as testing membership in a subgroup~\cite{MC87}.
Characterizing the complexity of this problem for common number systems, such as the integers, finite fields, or the integers modulo $k$ is therefore naturally of interest.
For an arbitrary constant $k \ge 2$, we consider the difficulty of \emph{deciding feasibility of linear congruences modulo $k$} (\LCON[k]) and \emph{computing solutions to linear congruences modulo $k$} (\LCONX[k]).
These are special cases of the problems \LCON\ and \LCONX\ defined by McKenzie and Cook~\cite{MC87}, in which $k \in O(n)$ is taken as part of the input and represented by its prime-power factors $p_1^{e_1} p_2^{e_2} \cdots p_\ell^{e_\ell}$.
Setting $k$ to a fixed constant is a natural, if slightly restrictive, special case.

Arvind and Vijayaraghavan~\cite{AV10} define the class $\Mod\Log \subset \NC^2$ as a logspace analogue the class \Mod\P\ defined by K\"obler and Toda~\cite{KT96}.
They show that \LCON\ is hard for \Mod\Log\ under \P-uniform $\NC^1$ reductions, and contained in $\Log^{\Mod\Log}/\Poly = \Log^{\num\Log}/\Poly$.
This is of course in contrast to the problem of determining integer-feasibility of integer matrix equations, which is at least as hard as computing greatest common divisors over $\Z$; the latter problem is not known to be in $\NC^j$ for any $j \ge 0$.

Buntrock~\etal~\cite{BDHM92} show --- for the special case of $k$ prime --- that determining the feasibility of systems of linear equations is complete for the complexity classes \Mod[k]\Log\ which generalize $\oplus\Log$, decidable by logspace nondeterministic Turing machines which can distinguish between having a number of accepting paths which are either zero or nonzero \emph{mod $k$}.
These results together with those of Ref.~\cite{AV10} suggest that the difficulty of solving linear equations over integer matrices is sensitive to the presence and the prime-power factorization of the modulus involved; one might suppose that \LCON[k] is particularly tractable for arbitrary $k \ge 2$.
For $k$ prime, Ref.~\cite{BDHM92} also shows that $\Mod[k]\Log = \co\Mod[k]\Log$: the techniques of Ref.~\cite{BDHM92} may more naturally be interpreted as proving that for $k$ prime, \LCON[k] is complete for \co\Mod[k]\Log.
Also implicit in Ref.~\cite{BDHM92} is that \LCON[k] is \co\Mod[k]\Log-hard for all $k \ge 2$ under $\NC^1$ reductions.
This suggests the question: for an \emph{arbitrary} modulus $k$, what is the precise relationship of \LCON[k] to the classes \co\Mod[k]\Log?


We show that the proof $\LCON \in \NC^3$ by McKenzie and Cook~\cite{MC87} may be adapted prove $\LCON[k] \in \co\Mod[k]\Log$, using fast parallel algorithms for matrix multiplication and computing rank modulo the prime factors of the modulus.
For a constant prime modulus $p$, the latter problems are complete for \co\Mod[p]\Log: however, as is typical of counting problems, they are ``evaluated'' in the number of accepting branches of the computation, which is an obstacle to performing operations such as integer division required by the McKenzie--Cook algorithm.
We overcome this obstacle by describing a class \FUL[k] of machines which evaluate functions on the work-tape, and which may be simulated in mod-logspace computations.\footnote{%
	For logspace nondeterministic machines, we adopt the Russo--Simon--Tompa oracle model~\cite{RST-1984}, in which nondeterministic machines are not allowed to make nondeterministic transitions while it writes on the oracle tape (\ie~oracle queries must be determined by the contents of the input and work tapes before the query has started being written).
}
This simulation uses techniques similar to those used Buntrock~\etal~\cite{BDHM92} for $p$ prime, to show closure of the class $\Mod[p]\Log$ under $\NC^1$ reductions.
We then describe, for $p^e$ a prime power, a \FUL[p^e] algorithm to solve the problem \LCONNULL[p^e] of computing a spanning set for a basis of the nullspace of a matrix modulo $p^e$.
It follows that \LCON[k] is \co\Mod[k]\Log-complete, and both \LCONX[k] and \LCONNULL[k] are \Fundot\co\Mod[k]\Log-complete (this class being the functional analogue of \co\Mod[k]\Log), for any constant $k \ge 2$.

Finally, for arbitrary moduli $k$, we consider the relationship of the class \FUL[k] to the function class \Fundot\co\Mod[k]\Log, and consider what insights it may suggest for oracle closure results of the form $\Mod[k]\Log^{\!\!\;\Mod[k]\Log} = \Mod[k]\Log$ for $k$ composite, where this problem remains open; and to consider what light it sheds on attempts to resolve it~\cite{Sz99}.

\section{Preliminaries}
\label{sec:preliminaries}

Throughout the following, $k \ge 2$ is a constant modulus, with a factorization into powers of distinct primes $k = p_1^{e_1} p_2^{e_2} \cdots p_\ell^{e_\ell}$.
When we consider the case of a modulus $k$ which is a prime power, we typically write $p^e$ instead, for $p$ some prime and $e \ge 1$ some positive integer.
We will suppose that the reader is familiar with the basic properties of the function classes \num\Log~\cite{BDHM92} and \Gap\Log~\cite{AO96}.

We consider the complexity of the following problems, which are named in analogy to problems considered by McKenzie and Cook~\cite{MC87}:
For an $m \x n$ integer matrix $A$ and vector $\vec y \in \Z^m$ provided as input, we define the following problems:
\begin{problems}
\item \textbf{\LCON[k]} ---
  Determine whether $A \vec x \equiv \vec y \pmod{k}$ has solutions for $\vec x \in \Z^n$.
\item \textbf{\LCONX[k]} ---
	Output a solution to the congruence $A \vec x \equiv \vec y \pmod{k}$, if one exists.
\item \textbf{\LCONNULL[k]} ---
	Output vectors $\vec x_1, \ldots, \vec x_N$ which span the solutions to $A \vec x \equiv \vec 0 \pmod{k}$.
\end{problems}
\noindent Without loss of generality, we may suppose $m = n$ by padding the matrix $A$.
We wish to describe how these problems relate to the classes \co\Mod[k]\Log\ for $k \ge 2$, which are the complements of the classes \Mod[k]\Log\ defined by Buntrock~\etal~\cite{BDHM92}.
(Because the classes \co\Mod[k]\Log\ are our principal concern, and because of the techniques used for completeness results in Ref.~\cite{BDHM92}, we will present the preliminary definitions and results which we use in terms of these classes, rather than the complementary classes \Mod[k]\Log.)

\begin{definition}
	The class \co\Mod[k]\Log\ (respectively \Mod[k]\Log) is the set of languages $L$ for which there exists $\varphi \in \num\Log$ such that $x \in L$ if and only if $\varphi(x) \equiv 0 \bmod{k}$ (respectively, $\varphi(x) \not\equiv 0 \bmod{k}$).
\end{definition}
\noindent
Note that the classes \Mod[k]\Log\ and \co\Mod[k]\Log\ remain the same if we substitute \num\Log\ in the definition above with \Gap\Log, as any function $g = f_1 - f_2 \in \Gap\Log$, for functions $f_1, f_2 \in \num\Log$, is congruent modulo $k$ to $f_1 + (k-1) f_2 \in \num\Log.$
The following results are a synopsis of (the remark which follows) Ref.~\cite[Theorem~10]{BDHM92}:
\begin{prop}
	\label{prop:coMod-det-itmat}
  We may characterize \co\Mod[k]\Log\ as the class of decision problems which are (logspace-uniform) $\NC^1$-reducible to verifying matrix determinants mod $k$, or verifying coefficients of integer matrix products/inverses mod $k$.
	(The corresponding \emph{falsification} problems are complete for \Mod[k]\Log.)
\end{prop}
\begin{prop}
	\label{prop:LCONprime}
  For $p$ prime, \LCON[p] is complete for \co\Mod[p]\Log\ under $\NC^1$ reductions.
\end{prop}
\noindent
Buntrock~\etal\ also characterize the classes \co\Mod[k]\Log\ in terms of the prime factors of $k$, and show closure results which will prove useful.
The following are equivalent to Lemma~6, Theorem~7, and Corollary~8 of Ref.~\cite{BDHM92} via logical complementation:
\begin{prop}[normal form]
	\label{prop:normalForm}
  Let $k = p_1^{e_1} p_2^{e_2} \cdots p_\ell^{e_\ell}$ be the factorization of $k \ge 2$ into prime powers $p_j^{e_j}$.
	Then $L \in \co\Mod[k]\Log$ if and only if there are languages $L_j \in \co\Mod[p_j]\Log$ such that $L = L_1 \inter \cdots \inter L_\ell$.
	In particular, $\co\Mod[k]\Log = \co\Mod[p_1 p_2 \cdots p_\ell]\Log$.
\end{prop}
\begin{prop}[closure under intersections]
  \label{prop:closedIntersection}
	For any $k \ge 2$ and languages $L, L' \in \co\Mod[k]\Log$, we have $L \inter L' \in \co\Mod[k]\Log$.
\end{prop}
\begin{prop}[limited closure under complements]
  \label{prop:complements}
	For any prime $p$ and $e \ge 1$, we have $\co\Mod[p^e]\Log = \co\Mod[p]\Log = \Mod[p]\Log = \Mod[p^e]\Log$.
\end{prop}
\noindent
By the Chinese Remainder Theorem, a system of linear congruences mod~$k$ has solutions if and only if it has solutions modulo each prime power divisor $p_j^{e_j}$ of $k$.
We then have $\LCON[k] \in \co\Mod[k]\Log$ if and only if $\LCON[p^e] \in \co\Mod[p^e]\Log = \co\Mod[p]\Log$ by Proposition~\ref{prop:normalForm}.
(In fact, this suffices to show that $\LCON[k] \in \co\Mod[k]\Log$ for all $k$ square-free.)

We see from Propositions~\ref{prop:LCONprime} and~\ref{prop:complements} that the case of a prime modulus is special.
For $p$ prime, Buntrock~\etal\ also implicitly characterize the complexity of \LCONX[p] and \LCONNULL[p]\,.
We may describe the complexity of these function problems as follows.
For a function $f(x): \Sigma^\ast \to \Sigma^\ast$ and $x \in \Sigma^\ast$, let $|f(x)|$ denote the length of the representation of $f(x)$; and let $f(x)_j$ denote the $j\textsuperscript{th}$ symbol in that representation.
Following Hertrampf, Reith and Vollmer~\cite{HRV00}, for a function $f: \Sigma^\ast \to \Sigma^\ast$ on some alphabet $\Sigma$, and for some symbol $\bullet \notin \Sigma$, we may define the decision problem
\begin{equation}
		\bits[f] = \left\{ (x,j,b) \;\left|\;
			\begin{array}{r@{}l@{~\text{and}~}r@{}l}
				\text{either}~	j&{}\le |f(x)|	&	b&{}=f(x)_j		\\
				\text{or}~			j&{}>|f(x)|			&	b&{}= \bullet
			\end{array}
			\right\}\right..
\end{equation}
Abusing notation, we write $f(x)_j = \bullet$ in case $|f(x)| < j$.
We extend this definition to \emph{partial} functions $f$ by asserting $(x,j,b) \in \bits[f]$ only if $x \in \dom(f)$.
\begin{definition}
	The class $\Fundot\co\Mod[k]\Log$ is the set of (partial) functions $f$ such that $|f(x)| \in \poly(|x|)$ for all $x \in \dom(f)$, and for which $\bits[f] \in \co\Mod[k]\Log$. (We define the class \Fun\Mod[k]\Log\ similarly.)
\end{definition}
\noindent Then Ref.~\cite[Theorem~10]{BDHM92} also implicitly shows:
\begin{prop}
	\label{prop:LCONX-LCONNULL-FunMod}
  For $p$ prime, the problems \LCONX[p] and \LCONNULL[p] are complete for $\Fun\Mod[p]\Log = \Fundot\co\Mod[p]\Log$ under $\NC^1$ reductions.
\end{prop}

\section{Natural function classes for modular logspace}
\label{sec:modLogspaceFunClasses}

We introduce two complexity classes in logarithmic space: a modular analogue of \num\Log, and a class of function problems which is naturally low for \Mod[k]\Log\ and \co\Mod[k]\Log.
We describe the relationships of these classes to \Fun\Mod[k]\Log\ and \Fundot\co\Mod[k]\Log, and to each other in the case of a prime modulus.
Some of these results may be regarded as encapsulating known techniques; we present them explicitly to simplify the presentation of the main results of the article.

The first class we define is a logspace variant of the class \num[k]\P\ described by Valiant~\cite[page~193]{Valiant79}:

\begin{definition}
	\label{def:numkLog}
  The function class $\num[k]\Log$ is the set of functions $f: \Sigma^\ast \to \{0,1,\ldots,k-1\}$ such that there exists $\varphi \in \num\Log$ such that $f(x) \equiv \varphi(x) \pmod{k}$.
\end{definition}

\noindent
Note that \num[k]\Log\ is closed under addition, multiplication, and constant powers \emph{modulo $k$} by virtue of similar closure results for \num\Log\ over the integers; it is closed under subtraction mod $k$ as well, as $M-N \equiv M + (k-1)N \pmod{k}$.
Thus, if we were to define a similar class \Gap[k]\Log\ in terms of congruence mod $k$ to functions $g \in \Gap\Log$, we would obtain $\Gap[k]\Log = \num[k]\Log$.
Note that by its definition, the functions $f \in \num[k]\Log$ are essentially those functions whose values a \co\Mod[k]\Log\ algorithm can verify directly:
\begin{lemma}
	\label{lemma:numkLogInFun}
  For any $k \ge 2$, $\num[k]\Log \subset \Fundot\co\Mod[k]\Log$.
\end{lemma}
\begin{proof}
  For $f \in \num[k]\Log$ such that $f: \Sigma^\ast \to \{0,\ldots,k-1\}$, let $\mathbf T$ be a nondeterministic Turing machine which accepts on inputs $x$ with some number of branches $\varphi(x) \equiv f(x) \pmod{k}$.
	Consider a nondeterministic logspace machine $\mathbf T'$ acting on the alphabet $\Sigma^\ast \cup \{0,\ldots,k-1\}$ which takes input tuples $(x,y) \in \Sigma^\ast \x \{0,\ldots,k-1\}$.
	The machine $\mathbf T'$ reads $y$ and branches $(k-1)y + 1$ times.
	In one of these branches, $\mathbf T'$ simulates $\mathbf T$ on $x$, accepting if and only if $\mathbf T$ does; in the other branches, it accepts unconditionally.
	The number of accepting branches is then $\varphi(x) + (k-1)y \equiv f(x) - y \pmod{k}$, so that $\mathbf T'$ accepts with $0$ branches mod $k$ on input $(x,y)$ if and only if $f(x) = y$.
	Thus $\bits[f] \in \co\Mod[k]\Log$.
\end{proof}
\noindent
The technique here is identical to that of Buntrock~\etal~\cite{BDHM92}; one might describe Ref.~\cite[Theorem~10]{BDHM92} as showing that evaluating matrix determinants modulo $k$, and evaluating coefficients of products/inverses of integer matrices modulo $k$, are contained in \num[k]\Log.

We are interested in logspace machines which compute \num[k]\Log\ functions on their output tapes.
We will be interested in a particular sort of nondeterministic logspace machine which is suitable for performing computations as subroutines of \co\Mod[k]\Log\ machines: the main result of this section is to describe conditions under which it can compute functions in \num[k]\Log.
\begin{definition}
	\label{def:FUL}
	A \emph{\FUL[k] machine computing a (partial) function $f$} is a nondeterministic logspace Turing machine which 
	\textbf{(a)}~for inputs $x \in \dom(f)$, computes $f(x)$ on its output tape in some number $\varphi(x,f(x)) \equiv 1 \pmod{k}$ of its accepting branches, and \textbf{(b)}~for each $y \ne f(x)$ (or for any string $y$, in the case $x \notin \dom(f)$), computes $y$ on its output tape on some number $\varphi(x,y) \equiv 0 \pmod{k}$ of its accepting branches.
	We say that $f \in \FUL[k]$ if there exists a \FUL[k] machine which computes $f$.
\end{definition}
If we replace the relation of equivalence modulo $k$ with equality in the definition of \FUL[k] above, we obtain the class \FUL\ of functions computable by nondeterministic logspace machines with a single accepting branch.
This is in turn analogous to the class $\mathsf{UPF}$ described by Beigel, Gill, and Hertrampf~\cite{BGH90}, of functions which may be computed by a nondeterministic polynomial time Turing machine without affecting the number of accepting branches of that machine.
Note that for a \FUL[k] machine $\mathbf U$, what is written on the output tape in many branches (perhaps even the vast majority of them) may not be the function $f(x)$ which $\mathbf U$ ``computes''; but as any string $y \ne f(x)$ occurs with multiplicity a multiple of $k$, the branches containing such $y$ cannot affect the number of accepting branches modulo $k$ of any machine simulating $\mathbf U$ as a subroutine.
In a \Mod[k]\Log\ or \co\Mod[k]\Log\ algorithm, such ``incorrect results'' occur in effect with measure zero. 

In this sense, the closure result $\Mod[p]\Log^{\Mod[p]\Log} = \Mod[p]\Log$ for $p$ prime which is implicit in Ref.~\cite{BDHM92} and explicitly shown in Ref.~\cite{HRV00} may be interpreted as saying that the characteristic function of any $L \in \Mod[p]\Log$ may be computed by a \FUL[p] machine.
That is, a \Mod[p]\Log\ oracle can be directly simulated in a \Mod[p]\Log\ algorithm, by simulating the corresponding \FUL[p] machine as a subroutine.
Our interest in the function class \FUL[k] is for essentially the same reason, \ie~an oracle for computing any function $f \in \FUL[k]$ can be substituted with a simulation of the \FUL[k] machine itself in the same manner:
\begin{lemma}
	\label{lemma:oracleClosure}
	For all $k \ge 2$\,, the class \FUL[k] is low for \Mod[k]\Log, \co\Mod[k]\Log, \num[k]\Log, and itself.
\end{lemma}
\noindent
The proof is trivial: one simply simulates the \FUL[k] machine computing $f$ as part of the nondeterministic logspace machine/algorithm for the corresponding decision/function class.
From simple number-theoretic considerations, the classes \FUL[k] have properties which are similar in appearance to those of \co\Mod[k]\Log\ (while in fact giving rise to opposite class containments):
\begin{theorem}
	\label{thm:FULnormalForm}
  Let $k = p_1^{e_1} p_2^{e_2} \cdots p_\ell^{e_\ell}$ be the factorization of $k \ge 2$ into prime power factors $p_j^{e_j}$.
	Then $\FUL[k] = \FUL[p_1] \inter \FUL[p_2] \inter \cdots \inter \FUL[p_\ell]$, and in particular $\FUL[k] = \FUL[p_1 p_2 \cdots p_\ell]$.
\end{theorem}
\begin{proof}
	Throughout the following, let $\kappa = p_1 p_2 \cdots p_\ell$ be the largest square-free factor of $k$.
	We first show $\FUL[\kappa] = \FUL[p_1]\! \inter \cdots \inter \FUL[p_\ell]$.
	Suppose $f \in \FUL[p_j]$ for each $1 \le j \le \ell$, and is computed by some \FUL[p_j] machine $\mathbf U_j$ in each case.
	Let
	\begin{equation}
			\gamma	\;=\;	\kappa/p_1 + \kappa/p_2 + \cdots + \kappa/p_\ell \;.
	\end{equation}
	For each prime $p_j$, all terms in the right-hand sum except for the $j\textsuperscript{th}$ term are divisible by $p_j$: then $\gamma$ has no prime divisors in common with $\kappa$, so that $\gcd(\gamma,\kappa) = 1$.
	Let $\beta \equiv \gamma^{-1} \pmod{\kappa}$, and consider the machine $\mathbf U'$ which performs the following:
	\begin{enumerate}
	\item
		Nondeterministically write some index $1 \le j \le \ell$ on the work tape.
	\item
		For each such $j$, nondeterministically select some integer $0 \le q < \kappa\beta/p_j$.
	\item
		In each branch, simulate $\mathbf U_j$ on the input $x$, accepting if and only if $\mathbf U_j$ accepts. 
	\end{enumerate}
	For any string $y \in \Sigma^\ast$ different from $f(x)$, the number of branches in which $\mathbf U_j$ accepts is $m_j p_j$ for some $m_j \in \N$; and so $\mathbf U'$ has $m_j \kappa \beta $ branches where $j$ is written on the work tape and $y$ is written on the output tape.
	Summing over all $j$, we find that any $y \ne f(x)$ is written on the output tape in a number of branches which is a multiple of $\kappa$.
	Similarly, for the case $y = f(x)$, the number of branches in which $\mathbf U_j$ accepts is $m_j p_j + 1$ for some $m_j \in \N$; and so $\mathbf U'$ has $m_j \kappa \beta + \kappa \beta/p_j$ branches where $j$ is written on the work tape and $f(x)$ is written on the output tape.
	Summing over all $j$ and neglecting multiples of $\kappa$, we have $\beta\bigl(\kappa/p_1 + \cdots + \kappa/p_\ell) = \beta\gamma \equiv 1 \pmod{\kappa}$ branches in which $f(x)$ is written on the output tape; thus $\mathbf U'$ is an \FUL[\kappa] machine computing $f$.
	The converse containment $\FUL[\kappa] \subset \FUL[p_j]$ for each $1 \le j \le \ell$ is trivial.

	It remains to show that $\FUL[\kappa] \subset \FUL[k]$, the reverse containment again being easy.
  Let $f \in \FUL[\kappa]$ be computed by an \FUL[\kappa] machine $\mathbf U'$ and have outputs of length bounded by $N := N(x) \in \poly(|x|)$.

	The idea of our approach is based on the following construction for $N \in O(\log |x|)$, which would for instance apply if we wished to evaluate logarithmically many symbols of $f(x)$ in the in the work-tape of another machine.
	We construct a \FUL[k] machine ${\mathbf U}''$ which computes $f$ by simply performing $k/\kappa$ consecutive independent simulations of $\mathbf U'$, recording the outcome of each simulation on the work tape.
	For each $1 \le j \le k/\kappa$, in any given computational branch, let $\sigma_j(x)$ be the string computed by the $j\textsuperscript{th}$ simulation of $\mathbf U'$.
	If any of the simulations reject the input, or produces a different output from the first simulation (\ie~if $\sigma_j(x) \ne \sigma_1(x)$ for any $1 \le j \le k/\kappa$), then $\mathbf U'$ rejects.
	Otherwise, $\mathbf U'$ writes the string $\sigma_1(x)$ agreed upon by the simulations to the output tape.

	The detailed analysis for $N \in \omega(\log |x|)$ proceeds by performing a similar simulation for blocks of output characters of some length $L := L(x) \in O(\log |x|)$.
	For each $1 \le m \le N/L$, define a machine $\mathbf U'_m$ which simulates $\mathbf U'$ except that it only writes the $m\textsuperscript{th}$ block of $L$ consecutive characters from $f(x)$, padding the end of $f(x)$ with a symbol $\bullet \notin \Sigma$ if necessary to obtain a string of length $N$.
	Let $M = N/L$ for the sake of brevity: rather than perform $k/\kappa$ simulations of $\mathbf U'$, the machine $\mathbf U''$ performs $k/\kappa$ simulations of each $\mathbf U'_m$ for $1 \le m \le M$, in sequence.
	Again, while simulating $\mathbf U'_m$ for any particular $m$, the machine $\mathbf U''$ stops and rejects if any of the simulations reject or produce a result inconsistent with the previous simulations; and in those branches in which $\mathbf U'_m$ has produced the same output $y\sur{m}$ each time, the string $y\sur{m}$ is written on the output tape (excluding any symbol $\bullet \notin \Sigma$).
	After finishing the simulations of $\mathbf{U}'_m$ for any $m < M$, it reuses the work-space to simulate the next machine $\mathbf{U}'_{m+1}$.
	Once the simulations of $\smash{\mathbf U'_M}$ are finished, $\mathbf U''$ accepts unconditionally in any branch where it has not yet rejected.

	Let $\varphi(x,y)$ be the number of computational branches in which $\mathbf U'$ accepts with the string $y \in \Sigma^\ast$ written on the tape: by hypothesis, $\varphi(x,y) \equiv 0 \pmod{\kappa}$ for each $y \ne f(x)$, and $\varphi(x,f(x)) \equiv 1 \pmod{\kappa}$.
	Similarly, let $\varphi_m(x,y\sur{m})$ be the number of branches in which $\mathbf U'_m$ accepts with $y\sur{m} \in (\Sigma \cup \{\bullet\})^L$ written on the tape for each $1 \le m \le M$, and $\Phi(x,y)$ be the number of branches in which $\mathbf U''$ accepts with $y \in \Sigma^\ast$ written on the tape.
	If $y \approx y\sur{1} y\sur{2} \cdots y\sur{M}$ (modulo any $\bullet$ symbols contained in any of the substrings $y\sur{m}$), then the number of branches in which $\mathbf U''$ accepts with a given string $y\sur{m}$ written on the $m\textsuperscript{th}$ block of $L$ tape cells of the output tape is \textbf{(a)}~independent of the other substrings $y\sur{j}$ for $j \ne m$, and \textbf{(b)}~is the result of $k/\kappa$ simulations of $\mathbf U'_m$ which each produce the substring $y\sur{m}$ as output; so that we have
	\begin{equation}
			\Phi(x,y)	\;=\;		\varphi_1\bigl(x,y\sur{1}\bigr)^{k\!\!\;/\!\!\:\kappa} \; \varphi_2\bigl(x,y\sur{2}\bigr)^{k\!\!\;/\!\!\:\kappa} \;\cdots \; \varphi_M\bigl(x,y\sur{M}\bigr)^{k\!\!\;/\!\!\:\kappa}.
	\end{equation}
	Note that $\varphi_m(x,y\sur{m})$ is equal to the number of computational branches in which $\mathbf U'$ writes any string $\sigma \in \Sigma^\ast$ on the output tape, for which the $m\textsuperscript{th}$ block is similar to $y\sur{m}$ (again ignoring any $\bullet$ symbols in $y\sur{m}$).
	This is the sum of $\varphi(x,\sigma)$ over all strings $\sigma$ consistent with the substring $y\sur{m}$.
	By hypothesis, $\varphi(x,\sigma)$ is a multiple of $\kappa$ except for the single case where $\sigma = f(x)$, in which case $\varphi(x,\sigma) \equiv 1 \pmod{\kappa}$.
	Thus $\varphi_m(x,y\sur{m}) \equiv 1 \pmod{\kappa}$ if $y\sur{m} \in \Sigma^L$ is consistent with the $m\textsuperscript{th}$ block of $f(x)$; otherwise, $\varphi_m(x,y\sur{m}) \equiv 0 \pmod{\kappa}$.
	We then observe the following:
	\begin{itemize}
	\item
		Let $E := \max \{\, e_j \mid k/p_j^{e_j} \in \Z \,\}$; then $E \le p_j^{E-1} \le k/\kappa$ for any $1 \le j \le \ell$.
		As $k$ divides $\kappa^E = p_1^E \cdots p_\ell^E \le \kappa^{k/\kappa}$, we then have $\varphi_m(x,y\sur{m})^{k/\kappa} \equiv 0 \pmod{k}$ if $\varphi_m(x,y\sur{m}) \equiv 0 \pmod{\kappa}$.
	\item
		The integers which are congruent to $1$ modulo $\kappa$ form a subgroup of order $k/\kappa$ within the integers modulo $k$; it then follows that $\varphi_m(x,y\sur{m})^{k/\kappa} \equiv 1 \pmod{k}$ if $\varphi(x,y\sur{m}) \equiv 1 \pmod{\kappa}$.
	\end{itemize}
	Taking the product over $1 \le m \le M$, we have $\Phi(x,y) \equiv 0 \pmod{k}$ unless each substring $y\sur{m}$ is consistent with the $m\textsuperscript{th}$ block of $f(x)$, in which case $y = f(x)$ and $\Phi(x,y) \equiv 1 \pmod{k}$.
	Thus $\mathbf U''$ is an \FUL[k] machine computing $f$.
\end{proof}

The requirement that an \FUL[k] machine have one accepting branch mod $k$ (or possibly zero if the machine computes a partial function) gives the following relation to the classes \Fun\Mod[k]\Log\ and \Fundot\co\Mod[k]\Log:
\begin{lemma}
	\label{lemma:FULinFun}
  For all $k \ge 2$, we have $\FUL[k] \subset \Fun\Mod[k]\Log \,\inter\, \Fundot\co\Mod[k]\Log$.
\end{lemma}
\begin{proof}
  Let $\mathbf U$ be a \FUL[k] machine computing $f: \Sigma^\ast \rightharpoonup \Sigma^\ast$.
	Consider a nondeterministic logspace machine $\mathbf T$ taking inputs $(x,j,b) \in {\Sigma^\ast \x \N \x \bigl(\Sigma \cup \{\bullet\}\bigr)}$, and which simulates $\mathbf U$, albeit ignoring all instructions to write to the output tape, except for the $j\textsuperscript{th}$ symbol which it writes to the work-tape.
	(If $j > |f(x)|$, $\mathbf T$ instead writes ``$\bullet$'' to the work-tape.)
	Then $\mathbf T$ compares the resulting symbol $f(x)_j$ against $b$, accepting if they are equal and rejecting otherwise.
	Then the number of accepting branches is equivalent to $1$ modulo $k$ if $f(x)_j = b$, and is a multiple of $p$ otherwise, so that $\bits[f] \in \Mod[k]\Log$.
	To show $\bits[f] \in \co\Mod[k]\Log$, we may consider a machine $\mathbf T'$ which differs from $\mathbf T$ only in that it rejects if $f(x)_j = b$, and accepts otherwise.
	Thus $\FUL[k] \subset \Fun\Mod[k]\Log \,\inter\, \Fundot\co\Mod[k]\Log$.
\end{proof}
\noindent 
This identifies $\FUL[k]$ as an important subclass of the existing logspace-modular function classes.
For prime-power moduli, we may sharpen Lemma~\ref{lemma:FULinFun} to obtain a useful identity:
\begin{lemma}
	\label{lemma:FULprimeEqualsFun}
	For any prime $p$ and $e \ge 1$, $\FUL[p^e] = \Fun\Mod[p]\Log = \Fundot\co\Mod[p]\Log$.
\end{lemma}
\begin{proof}
	By Proposition~\ref{prop:complements}, Theorem~\ref{thm:FULnormalForm}, and Lemma~\ref{lemma:FULinFun}, it suffices to prove $\Fun\Mod[p]\Log \subset \FUL[p]$ for $p$ prime.
	For $f \in \Fun\Mod[p]\Log$, let $\mathbf T$ be a \Mod[p]\Log\ machine which decides \bits[f].

	We construct a family of machines $\mathbf T_{j,b}$ (for each $j \in \N$ and $b \in \Sigma \cup \{\bullet\}$), where each machine $\mathbf T_{j,b}$ writes $b$ on its output tape and simulates $\mathbf T$ to decide whether $(x,j,b) \in \bits[f]$ on an input $x \in \Sigma^\ast$.
	Without loss of generality, as in Ref.~\cite[Corollary~3.2]{HRV00} each machine $\mathbf T_{j,b}$ accepts on a number of branches $\varphi(x,j,b) \equiv 1 \pmod{p}$ if case $f(x)_j = b$, and $\varphi(x,j,b) \equiv 0 \pmod{p}$ otherwise.

	We form a \FUL[p] machine $\mathbf U_j$ computing $f(x)_j$ by taking the ``disjunction'' of the machines $\mathbf T_{j,b}$ over all $b \in \Sigma \cup \{\bullet\}$: \ie~$\mathbf U_j$ branches nondeterministically by selecting $b \in \Sigma \cup \{\bullet\}$ to write on the work-tape and simulates $\mathbf T_{j,b}$, accepting with one branch mod $p$ if and only if $b = f(x)_j$ and accepting with zero branches mod $p$ otherwise.
	Given some upper bound $|f(x)| \le N(x) \in \poly(|x|)$, we then construct a \FUL[p] machine $\mathbf U$ to compute $f(x)$ by simply simulating $\mathbf U_j$ for each $1 \le j \le N(x)$ in sequence, writing the symbols $f(x)_j$ individually on the output tape; accepting once it either computes a symbol $f(x)_j = \bullet$ (without writing $\bullet$ to the output) or the final iteration has been carried out.
\end{proof}
\noindent
This result is the crux of the result of Ref.~\cite{HRV00}, albeit extended beyond the characteristic functions of $L \in \Mod[p]\Log$: when $k$ is a prime power, \emph{any} function whose bits are verifiable by \co\Mod[k]\Log\ algorithms, can also be evaluated naturally as a subroutine of a \co\Mod[k]\Log\ algorithm.
The importance of this result to us lies in the consequence for \num[k]\Log, as the prototypical class of functions verifiable in \co\Mod[k]\Log:
\begin{corollary}
	\label{cor:numLogLow}
  $\num[p^e]\Log \subset \FUL[p]$ for any prime $p$ and $e \ge 1$.
	It follows that \num[p^e]\Log\ is low for $\Mod[p]\Log = \co\Mod[p]\Log$ and for $\num[p]\Log$ in this case.
\end{corollary}
\noindent
This follows from Proposition~\ref{lemma:numkLogInFun} and Lemma~\ref{lemma:oracleClosure}, and is the key technical ingredient of our result: it allows us to simulate logspace counting oracles modulo $p^e$ as a part of a \co\Mod[p]\Log\ algorithm.

\section{Solving congruences and nullspaces mod $k$}
\label{sec:reduction}

We return to the motivating problems of this article.
Let $A$ be an $n \x n$ integer matrix and $\vec y \in \Z^n$ be provided as the input to \LCON[k] or \LCONX[k]; or $B$ be an $n \x n$ matrix provided as input to \LCONNULL[k]\,.
Without loss of generality, the coefficients of $A$ and $\vec y$, or of $B$, are non-negative and bounded strictly above by $k$ (as reducing the input modulo $k$ can be performed in $\NC^1$).
We follow the analysis of Ref.~\cite[Section~8]{MC87} which reduces solving linear congruences to computing generating sets for nullspaces modulo the primes $p_j$ dividing $k$.
The contribution of this section is to show that the latter problem can be solved for prime powers via a reduction to matrix multiplication together with modular counting oracles from \num[p^e]\Log\ for prime powers $p^e$.

We consider nondeterministic logspace machines operating on an alphabet $\bar \Sigma_k = \Sigma_k \cup \{\bullet\}$, where $\Sigma_k = \{0,1,\ldots,k-1\}$.
For the function problems \LCONNULL[k] and \LCONX[k], we wish respectively to compute
\begin{itemize}
\item 
	a function $\mathcal N_k: \Sigma_k^{n^2} \to \Sigma_k^{Nn}$ for $N \in O(n)$ such that $\mathcal N_k(B)$ is a sequence of vectors $(\vec Z_0, \vec Z_1, \ldots, \vec Z_{N-1})$ which generate $\Null(B)$ in $\Z/k\Z$; and
\item
	a partial function $\mathcal S_k: \Sigma_k^{n^2+n} \rightharpoonup \Sigma_k^n$ such that $(A,\vec y) \in \dom(\mathcal S_k)$ if and only if there exists a solution $\vec x$ to the system $A \vec x \equiv \vec y \pmod{k}$, in which case $\vec x = \mathcal S_k(A,\vec y)$ is such a solution.
\end{itemize}

For $p$ prime, we first consider a logspace reduction to matrix inversion and iterated matrix products modulo $p^e$, in a machine equipped with a \num[p^e]\Log\ oracle to compute certain matrix coefficients.
The reduction itself is an adaptation of the analysis of McKenzie and Cook~\cite[Lemma~8.1]{MC87}, together with observations regarding the simulation of the \num[p^e]\Log\ oracles.

\begin{lemma}
	\label{lemma:LCONNULL-primePowers}
	For any $p$ prime and $e \ge 1$, we have $\LCONNULL[p^e] \in \FUL[p]$.
\end{lemma}
\begin{proof}
	We reduce \LCONNULL[p^e] to \LCONNULL[p]\,, matrix products, and access to oracles for computing coefficients of certain matrices.
	We proceed by showing, for each $1 \le t \le e$, that computing any individual coefficient from a set of vectors $\vec V\sur{t}_j$ which span $\Null(B)$ modulo $q = p^t$ can be achieved by a \num[p^t]\Log\ function.
	We remark on the case $t = 1$ further below, and suppose as an induction hypothesis that there exists some $1 \le t < e$ for which computing the coefficients of such a spanning set can be performed by \num[p^t]\Log\ functions.

	We sketch the analysis of Ref.~\cite[Lemma~8.1]{MC87} for completeness.
	Suppose that we have a generating set $\smash{\vec V\sur{t}_{\!1}, \ldots, \vec V\sur{t}_{\!N_t}}$ over $\Z/p^e\Z$ for the nullspace of $B$ modulo $p^t$, forming the columns of an $N_t \x n$ matrix $V\sur{t}$.
	Certainly any solution to $B \vec w \equiv 0 \pmod{p^{t+1}}$ must also be a solution to $B \vec w \equiv 0 \pmod{p^t}$;
	\begin{subequations}
	then we may decompose such $\vec w$ modulo $p^e$ as a linear combination of the vectors $\smash{\vec V\sur{t}_{\!j}}$,
	\begin{equation}
		\vec w = u_1 \vec V\sur{t}_{\!1} + \cdots + u_{N_t} \vec V\sur{t}_{\!N_t} + p^t \vec{\hat w}
	\end{equation}
	for some $\vec{\hat w} \in \Z^n$; or more concisely,
	\begin{equation}
			\vec w = \tilde V\sur{t} \vec z	\;,
	\end{equation}
	for 
	block matrices $\tilde V\sur{t} = \smash{\bigl[\, \vec V\sur{t}_{\!1} \,\, \vec V\sur{t}_{\!2} \; \cdots \; \vec V\sur{t}_{\!N_t} \;\big|\; p^t I \;\bigr]}$ and $\vec z = \smash{\bigl[\, u_1 \,\, u_2 \;\cdots\; u_{N_t}\bigr]}\trans \oplus \vec{\hat w} \in \Z^{N_t + n}$.
	\end{subequations}
	To consider the additional constraints imposed by $B\vec w \equiv 0 \pmod{p^{t+1}}$, consider a decomposition $B = B_t + p^t \hat B_t$, where the coefficients of $B_t$ are bounded between $0$ and $p^t$.
	\begin{subequations}
	We then have 
	\begin{align}
		\Biggl( \sum_{j=1}^{N_t} u_j \Bigl[ B_t \vec V\sur{t}_{\!j} \,&+\;p^t \!\hat{B}_t \vec V\sur{t}_{\!j} \Bigr] \Biggr) + p^t B_t \vec{\hat w} 
		\notag\\[-1ex]&\equiv\,
		B \Biggl( \sum_{j=1}^{N_t} u_j \vec V\sur{t}_{\!j} \!\Biggr) + p^t \vec{\hat w} \equiv 0 \pmod{p^{t+1}}\,.
	\end{align}
	As the coefficients of each $B_t \vec V\sur{t}_{\!j}$ is divisible by $p^t$ by construction, we may simplify to
	\begin{equation}
		\Biggl( \sum_{j=1}^{N_t} u_j \Bigl[ B_t \vec V\sur{t}_{\!j} \!\big/ p^t \,+\; \!\hat{B}_t \vec V\sur{t}_{\!j} \Bigr] \Biggr) + B_t \vec{\hat w} 
		\equiv 0 \pmod{p}\,,
	\end{equation}
	or somewhat more concisely,
	\begin{equation}
	\label{eqn:liftingConstraint}
		\bar B\sur{t} \vec z \equiv 0 \pmod{p},
	\end{equation}
	\end{subequations}
	where we define
	\begin{align}
	\label{eqn:defineBbar}
		\bar B\sur{t} &=\, \bigl[\, \vec b\sur{t}_{1} \,\, \vec b\sur{t}_{2} \; \cdots \; \vec b\sur{t}_{N_t} \,\bigr|\; B_t \,\bigr],
	&
		\quad\text{for}\;\;
		\vec b\sur{t}_{j} &=\, B_t \vec V\sur{t}_{\!j} \!\big/ p^t  \;+\; \hat{B}_t \vec V\sur{t}_{\!j}	,
	\end{align}
	and where $\vec z$ is as we defined it above.
	To find not just one vector $\vec w$ but a set of generators $\smash{\vec V\sur{t+1}_{\!1}\!},\, \ldots, \smash{\vec V\sur{t+1}_{\!N_{t+1}}}$ over $\Z/p^e\Z$ for $\Null(B)$ mod $p^{t+1}$, it suffices to find a generating set $\vec z_1, \ldots, \vec z_{N_{t+1}}\!$ for the nullspace of $\bar B\sur{t}$ mod $p$, and then set $\smash{\vec V\sur{t+1}_{\!h} = \tilde V\sur{t} \vec z_{h}}$.

	Note that the dimension of the nullspace of $\bar B\sur{t}$ modulo $p$, taken as a subspace of the vector space $\F_p$, is bounded by $N_t + n$; we may then span $\Null(B)$ mod~$p$ by vectors $\vec z_1 = p \unit_1$, $\vec z_2 = p \unit_2$, \ldots, $\vec z_{N_t} = p \unit_{N_t}$, and a collection of at most $N_t + n$ vectors $\vec z_h$ representing non-trivial vectors in $\Null(\bar B\sur{t})$ mod $p$ which have coefficients bounded between $0$ and $p$.
	We may take these vectors as the columns of a matrix $Z\sur{t+1}$: then we may compute a matrix whose columns span the nullspace of $B$ modulo $p^{t+1}$ as $V\sur{t+1} = \tilde V\sur{t} Z\sur{t+1}$.
	The number of columns of $V\sur{t+1}$ is $N_{t+1} \le 2N_t + n$, by construction.
	Thus we have a reduction of \LCONNULL[p^{t+1}] to evaluating $\bar B\sur{t}$, solving \LCONNULL[p] on input $\bar B\sur{t}$ to obtain the coefficients of the matrix $Z\sur{t+1}$, and matrix multiplication modulo $p^{t+1}$ of the matrices $\tilde V\sur{t}$ and $Z\sur{t+1}$.
	\begin{itemize}
	\item
		Coefficients of an iterated matrix product $M_1 M_2 \cdots M_{\poly(n)}$ modulo $k$ may be evaluated as \num[k]\Log\ functions, using the approach outlined in Ref.~\cite[Proposition~9]{BDHM92}.
		One simulates a branching program with nondeterministic choices, using the matrices in sequence as transition functions for each branching.
		If the coefficients of the matrices can be evaluated using an oracle for a class $C$ which is low for \num[k]\Log\ (such as $C = \FUL[k]$), the coefficients of the matrix product can then be straightforwardly computed as \num[k]\Log\ functions.

	\item
		Buntrock~\etal~\cite[Theorem~10]{BDHM92} implicitly show that individual coefficients of a spanning set for $\Null(B) \bmod{p}$ for integer matrices $B$ are \num[p]\Log\ functions (as in the remarks following Lemma~\ref{lemma:numkLogInFun}),
		using a sequence of $\NC^1$-reductions --- specifically those of Ref.\cite[Theorem~5]{BGH82} and Refs.~\cite{Berk84,Mulm87}, as well as conjunctive and disjunctive truth-table reductions which rely on Propositions~\ref{prop:closedIntersection} and~\ref{prop:complements}.
	\end{itemize}
	For the base case of $t = 1$, the latter observation immediately shows that coefficients of $\mathcal N_p$ may be computed by $\FUL[p]$ machines.
	To induct, let $q = p^{t+1}$, and suppose that for some $t \ge 1$ we have $\LCONNULL[p^t] \in \FUL[p] = \FUL[q]$.
	We may then evaluate the coefficients of a set of vectors $\mathbf V\sur{t}$ which span $\Null(B)$ modulo $p^t$ by simulating a \FUL[q] machine; coefficients of the matrix product with $B$ may then be computed in \num[q]\Log\ using the first observation above, which can be performed by a \FUL[q] machine using Corollary~\ref{cor:numLogLow}.
	The columns of $\bar B\sur{t}$ are either integer vectors of the form $B \vec V\sur{t}/p^t$, or are columns of $\hat B_t$: both can then be computed by \FUL[q] subroutines, as division by $p^t$ (which is bounded by the constant $p^e$) can be performed in $\NC^1$, and $B V\sur{t}$ can be computed by simulating \FUL[q] machines, and $\hat B_t$ is also obtained from the input matrix $B$ by integer division by $p^t$.
	Again using $\LCONNULL[p] \in \FUL[p] = \FUL[q]$, the coefficients of $Z\sur{t+1}$ are all either constant or in effect computable by \FUL[q] machines; the coefficients of $\tilde V\sur{t}$ are similarly constant or computable by \FUL[q] machines.
	We may then compute the coefficients of the matrix product $V\sur{t+1} = \tilde V\sur{t} Z\sur{t+1}$ as a $\num[q]\Log \subset \FUL[p]$ function.
	By induction up to $e$, we may then compute coefficients of $\mathcal N_{p^e}$ in $\FUL[p]$.
	
	To show that the entire function $\mathcal N_{p^e}$ may be computed in $\FUL[p]$, it suffices to bound the number of spanning vectors for the nullspace, to ensure that the matrices involved in the reductions are of polynomial size.
	By induction, the number of vectors $\vec V\sur{e}_j$ in the generating set will be $N_e \le n + 2n + \cdots + 2^{e-1}n \le p^e n \in O(n)$; it then follows that $\LCONNULL[p^e] \in \FUL[p]$.
\end{proof}

\noindent\textbf{A remark on oracle towers.}
The above reduction is recursive, but has constant depth, as $e \in O(1)$.
In particular, the exponent $e$ corresponds to the height of a tower of \FUL[p] oracles computing \num[p^t]\Log\ functions.
To simulate these oracles as part of \eg~a \co\Mod[p]\Log\ algorithm, the space resources can be described straightforwardly using a stack model of the work tape:
each nested \num[p^t]\Log\ oracle is simulated as a \FUL[p^e] subroutine which is allocated $O(\log |B|) = O(\log(n))$ space on the tape (where $|B| \in O(n^2)$ is the size of the input matrix after reduction modulo $p^e$), and which makes further recursive calls to \FUL[p^e] subroutines which do likewise, down depth at most $e$.
The space resources then scale as $O(e \log (n)) = O(\log(n))$.

Following Ref.~\cite[Lemma~5.3]{MC87}, we may reduce \LCON[k] and \LCONX[k] for $k \ge 2$ to \LCONNULL[k], as follows.
Suppose $A \vec x \equiv \vec y \pmod{k}$ has solutions.
Consider $B = [\,A\,|\, \vec y\,]$: then there are solutions $\bar{\vec x} = \vec x \oplus x_{n+1}$ to the equation $B \bar{\vec x} \equiv 0 \pmod{k}$ in which $x_{n+1} = -1$, and more generally in which $x_{n+1}$ is coprime to $k$.
Conversely, if there is such a solution $\bar{\vec x}$ to $B \bar{\vec x} \equiv 0 \pmod{k}$, we may take $\alpha \equiv -x_{n+1}^{-1} \pmod{k}$ and obtain $A (\alpha \vec x) \equiv  -\alpha x_{n+1} \vec y \equiv \vec y \pmod{k}$.
To determine whether $A \vec x \equiv \vec y \pmod{k}$ has solutions, or to construct a solution, it thus suffices to compute a basis for the nullspace of $B$, and determine from this basis whether any of the vectors $\bar{\vec x} \in \Null(B)$ have a final coefficient coprime to $k$; if so, the remainder of the coefficients of $\bar{\vec x}$ may be used to compute a solution to the original system.

\begin{lemma}
	\label{lemma:LCON-and-LCONX-primePowers}
	For any prime $p \ge 2$ and $e \ge 1$, we have $\LCON[p^e] \in \Mod[p]\Log$ and $\LCONX[p^e] \in \Fun\Mod[p]\Log$.
\end{lemma}
\begin{proof}
	We demonstrate an algorithm for both problems on a deterministic logspace machine with a $\FUL[p]$ oracle.
	For prime power moduli, $x_{n+1}$ is coprime to $p^e$ if and only if $p$ does not divide $x_{n+1}$.
	To solve \LCON[p^e] and \LCONX[p^e], we compute individually the final coefficients of the vectors $(\vec Z_0,\vec Z_1,\vec Z_2,\ldots) = \mathcal N_{p^e}(B)$ for $B = [\,A\,|\, \vec y\,]$, searching for an index $1 \le h \le N_e$ for which the dot product $\unit_{n+1} \cdot \vec Z_h$ is not divisible by $p$.
	Without loss of generality, we select the minimum such $h$.
	Finding such an index, or determining that there are none, is feasible for $\Fun\Log^{\FUL[p]}$ by using the oracle to evaluate coefficients, and then deterministically testing divisibility.
	If there is no such index $h$, we indicate that no solution exists by rejecting unconditionally. 
	Otherwise, there exists a solution to the linear congruence.
	\begin{itemize}
	\item 
		To indicate that $(A,\vec y)$ is a \emph{yes} instance of \LCON[p^e], we simply accept.
	\item
		To solve \bits[\mathcal S_k], we store the relevant coefficient $x^{(h)}_{n+1}$ on the work tape in binary and compute $\alpha \equiv -x_{n+1}^{-1} \pmod{p^e}$ deterministically.
		Using the \FUL[p] oracle, we then query the coefficients $z_{h,j}$ of $\smash{\vec Z_h}$, and compare $\alpha z_{h,j}$ to input coefficients, accepting (to indicate a \emph{yes} instance) if and only if the coefficients match. \qedhere
	\end{itemize}
\end{proof}


We may use the above results, together with the Chinese Remainder Theorem and Proposition~\ref{prop:normalForm}, to show that \LCON[k], \LCONX[k], and \LCONNULL[k] are complete problems for \co\Mod[k]\Log\ and \Fundot\co\Mod[k]\Log\ respectively for arbitrary $k \ge 2$.
This is easiest for \LCON[k], but the same basic approach may be used in each case.

\smallskip
\noindent\textbf{A remark on $\NC^1$ reductions.} Our usage of the terminology of $(\co)\Mod[k]\Log$-completeness below follows that of Ref.~\cite{BDHM92} (Definition~5 and the introduction to Section~3). Note that in the case of $k$ composite, the classes $\Mod[k]\Log$ are not known to be closed under $\NC^1$ reductions (this would imply, for instance, that $\Mod[k]\Log$ is closed under complements, and in particular low for itself).
By Propositions~\ref{prop:closedIntersection} and~\ref{prop:normalForm} together with Ref.~\cite[Lemma~6(v)]{BDHM92}, we may show that $\co\Mod[k]\Log$ is closed under a particular kind of $\NC^1$ reduction, in which the final output gate is an \textsc{and} gate, and each subtree which produces an input to that gate is produced by an $\NC^1$ circuit with oracles for $\Mod[p_j]\Log$ for some \emph{single} prime $p_j$ which divides $k$ (the primes may vary for different subtrees).
In the results below one may substitute this sort of reduction for $\NC^1$ reductions, with no confusion.

\begin{theorem}
	\label{thm:lcon-complete}
  For all $k \ge 2$, $\LCON[k]$ is $\co\Mod[k]\Log$-complete under $\NC^1$ reductions.
\end{theorem}
\begin{proof}
	Let $k = q_1 q_2 \cdots q_\ell$ for powers of distinct primes $q_j = p_j^{e_j}$.
	As we implied in the introduction, one may reduce \LCON[p_j]\! to \LCON[k], for any prime $p_j$ dividing $k$, by considering the feasibility of the congruence
	\begin{equation}
		(kA/p_j)\, \vec x \;\equiv\; k\vec y / p_j \pmod{k},
	\end{equation}
	which is equivalent to $A \vec x \equiv \vec y \pmod{p_j}$, by dividing both sides and the modulus by $k/p_j \in \N$.
	By Propositions~\ref{prop:LCONprime} through~\ref{prop:closedIntersection}, all problems in \LCON[k] may be reduced to solving some instances of \LCON[p_j] for each ${1 \le j \le \ell}$: then \LCON[k] is \co\Mod[k]\Log-hard.
	Using the Chinese Remainder Theorem, we also have $\LCON[k] = \LCON[q_{\!\!\;1}] \inter \cdots \inter \LCON[q_{\!\!\;\ell}]$.
	As $\LCON[q_{\!\!\;j}] \in \Mod[p_j]\Log = \co\Mod[p_j]\Log$ for each $1 \le j \le \ell$, it follows by Proposition~\ref{prop:normalForm} that $\LCON[k] \in \co\Mod[k]\Log$ as well.
\end{proof}

\begin{theorem}
	\label{thm:LCONX-LCONNULL-complete}
  For all $k \ge 2$, \LCONX[k] and \LCONNULL[k] are $\Fundot\co\Mod[k]\Log$-complete under $\NC^1$ reductions.
\end{theorem}
\begin{proof}
	Let $k = q_1 q_2 \cdots q_\ell$ for powers of distinct primes $q_j = p_j^{e_j}$.
	We define \prob{congbits$(f,q_j)$} to be the decision problem of determining for inputs $(x,h,b) \in \Sigma_k^\ast \x \N \x \bar\Sigma_k$ whether $x \in \dom(f)$, and if so, whether either $f(x)_h \equiv b \pmod{q_j}$ for $b \ne \bullet$ or $f(x)_h = \bullet = b$.

\begin{itemize}
\item 
	Clearly \bits[\mathcal S_k] is the intersection of the problems \prob{congbits$(\mathcal S_k,q_j)$} for ${1 \le j \le \ell}$ by the Chinese Remainder Theorem.
	We show $\prob{congbits$(\mathcal S_k,q_j)$} \in \co\Mod[q_j]\Log$ for each ${1 \le j \le \ell}$, as follows.
	For $b \in \Sigma_k$, we may expand $b$ in binary on the work tape and evaluate its reduction $0 \le b' < q_j$ modulo a given prime power $q_j$; for $b = \bullet$ we simply let $b' = \bullet$ as well, so that $b' \in \bar\Sigma_{q_j}$.
	We perform a similar reduction for each coefficient in $(A,\vec y)$ to obtain an input $(A',\vec y')$ with coefficients in $\Sigma_{q_j}$.
	With a \co\Mod[p_j]\Log\ algorithm, we may then decide whether $((A',\vec y'),h,b') \in \bits[\mathcal S_{q_j}]$.
	Thus $\bits[\mathcal S_k] \in \co\Mod[k]\Log$.
\item
	We follow the reduction of Ref.~\cite[Theorem~8.3]{MC87}, to show $\bits[\mathcal N_k] \in \co\Mod[k]\Log$.
	Given vectors $\smash{\vec X\sur{q_j}_1, \ldots, \vec X\sur{q_j}_{N_j}}$ spanning the nullspace of $B$ modulo $q_j$ for each $1 \le j \le \ell$, the nullspace of $B$ modulo $k$ is spanned over the integers modulo $k$ by the vectors
	\begin{equation}
	  \tfrac{k}{q_1} \vec X\sur{q_1}_1,\;\ldots\,,\;\tfrac{k}{q_1} \vec X\sur{q_1}_{N_1},\; \tfrac{k}{q_2}\vec X\sur{q_2}_1,\;\ldots\,,\tfrac{k}{q_j}\vec X\sur{q_j}_h\,,\;\ldots\,,\;\tfrac{k}{q_\ell} \vec X\sur{q_\ell}_{N_\ell}\;.
	\end{equation}
	(We omit the vectors $k \unit_h$ included by Ref.~\cite{MC87}, as these are congruent to $\vec 0$ in $\Z/k\Z$.)
	Let $\vec Z_h$ be the list of such vectors, for $0 \le h < N_1 + \cdots + N_\ell$: we suppose without loss of generality that $\mathcal N_k$ is defined, for $k$ divisible by more than one prime, to produce this sequence of vectors as output.
	If we define
	\begin{align}
		M_j	\;=\; \sum_{t=1}^j N_t\;,
	\end{align}
	then each vector $\vec Z_h$ is congruent to $\vec 0$ modulo $q_j$, for every $j \ge 1$ such that $h < M_{j-1}$ or $h \ge M_j$.
	We may then reduce \bits[\mathcal N_k] to testing the congruence of coefficients of $\vec Z_h$ with $0$ modulo $q_j$ for all prime powers for which $h < M_{j-1}$ or $h \ge M_j$, and testing congruence with the coefficients of $\smash{\frac{k}{q_j} \vec X\sur{q_j}_{h - M_j+1}}$ otherwise.
	These congruences modulo each prime power $q_j$ can again be evaluated in \co\Mod[q_j]\Log\ algorithm for 
	\prob{congbits$_j(\mathcal N_k)$}, using the logspace reduction to \bits[\mathcal N_{q_j}] as above.
\end{itemize}
These suffice to show that
$\LCONX[k], \LCONNULL[k] \in \Fundot\co\Mod[k]\Log$ for all $k$.
To show that $\mathcal S_k$ is $\Fundot\co\Mod[k]\Log$-complete under logspace many-to-one reductions, we may note (as in the proof in Theorem~\ref{thm:lcon-complete} for \LCON[k]) that solving $\bits[\mathcal S_k]$ suffices to solve $\bits[\mathcal S_{p_j}]$ for each prime $p_j$ dividing $k$; as $\bits[\mathcal S_{p_j}]$ is \co\Mod[p_j]\Log-complete for each $1 \le j \le \ell$ by Ref.~\cite[Theorem~10]{BDHM92}, we may reduce any collection of languages $L_1, L_2, \ldots, L_\ell$ such that $L_j \in \co\Mod[p_j]\Log$ to $\bits[\mathcal S_k]$.
Then the intersection $L = L_1 \inter L_2 \inter \cdots \inter L_\ell \in \co\Mod[k]\Log$ is also logspace reducible to $\bits[\mathcal S_k]$.
As all languages $L \in \co\Mod[k]\Log$ have such a form by Proposition~\ref{prop:normalForm}, including all problems of the form $\bits[f]$ for $f \in \Fundot\co\Mod[k]\Log$.
It then follows that $\bits[\mathcal S_k]$ is \co\Mod[k]\Log\ complete; a similar result obtains for $\bits[\mathcal N_k]$.
\end{proof}

\section{Further Remarks}
\label{sec:remarks}

The above analysis was motivated by observing that the reduction of McKenzie and Cook~\cite{MC87} for \LCONX\ and \LCONNULL\ (which take the modulus $k$ as input, as a product of prime powers $p_j^{e_j} \in O(n)$) was very nearly a projective reduction to matrix multiplication, and that it remained only to find a way to realize the division by prime powers $p^t$ involved in the reduction to \LCONNULL[p].
By showing that logspace counting oracles modulo $p^e$ could be simulated by a \co\Mod[p]\Log\ machine, using the function class \FUL[k] as a notion of naturally simulatable oracles for the classes \Mod[k]\Log\ and \co\Mod[k]\Log, the containments of Theorems~\ref{thm:lcon-complete} and~\ref{thm:LCONX-LCONNULL-complete} became feasible.

In the recursive reduction for \LCONNULL[p^e], the fact that $e \in O(1)$ is essential not only for the logarithmic bound on the work tape, but also for the running time for the \co\Mod[k]\Log\ algorithm to be polynomial.
The \FUL[p] machines used to implement the \num[p^e]\Log\ oracles, from the constructions of Theorem~\ref{thm:FULnormalForm} and Lemma~\ref{lemma:FULprimeEqualsFun}, implicitly involve many repeated simulations of \co\Mod[p]\Log\ algorithms ($p^e / p = p^{e-1}$ times each) to decide equality of counting functions with residues $0 \le r < p^e$: this contributes to a factor of overhead growing quickly with $e$.
Therefore our results are mainly of theoretical interest, characterizing the complexity of these problems with respect to logspace reductions.
It is reasonable to ask if there is an algorithm on a \co\Mod[p]\Log\ machine for \LCONNULL[p^e], whose running time grows slowly with $e$.

We may use the classes $\num[k]\Log$ and $\FUL[k]$ to explore the consequences of other closure results.
For instance, if \Mod[k]\Log\ is closed under oracle calls, we may use the following alternative formulation of $\Mod[k]\Log$ to show that $\chi_k \in \FUL[k]$: 
\begin{lemma}
	\label{lemma:coMod-k-alt}
	For every $k \ge 2$, $L \in \co\Mod[k]\Log$ if and only if there exists $\varphi \in \num\Log$ such that $x \in L$ if and only if $\varphi(x)$ is \emph{coprime} to $k$.
	Furthermore, without loss of generality, $x \in L \implies \varphi(x) \equiv 1 \pmod{k}$.
\end{lemma}
\begin{proof}
  For $k = p_1^{e_1} p_2^{e_2} \cdots p_\ell^{e_\ell}$ as usual, we have $L \in \co\Mod[k]\Log$ if and only if $L = L_1 \inter L_2 \inter \cdots \inter L_\ell$ for languages $L_j \in \co\Mod[p_j]\Log = \Mod[p_j]\Log$ by Propositions~\ref{prop:normalForm} and~\ref{prop:complements}.
	Let $\mathbf T_1, \ldots, \mathbf T_\ell$ be nondeterministic logspace machines such that $\mathbf T_j$ accepts on input $x$ with a number of branches not divisible by $p_j$ if $x \in L_j$, and with zero branches modulo $p_j^{e_j}$ otherwise.
	Following Ref.~\cite{BDHM92}, we may without loss of generality suppose that $\mathbf T_j$ accepts $x \in L_j$ with a single branch modulo $p_j$.
	Using a similar construction to that of Theorem~\ref{thm:FULnormalForm} for the square-free case, we may obtain a \emph{single} nondeterministic logspace machine $\mathbf T$ which accepts on a single branch modulo $p_j$ if $x \in L_j$, and on a number of branches equivalent to $0$ mod $p_j$ otherwise.
	If $x \in L$, then the number of branches on which $\mathbf T$ accepts is equivalent to one modulo every prime $p_j$, which means that it is equivalent to one mod $k$; otherwise, there is some prime $p_j$ which divides the number of accepting branches, so that the number of branches is not coprime to $k$.
\end{proof}
\noindent
From this characterization of $\co\Mod[k]\Log$, we may argue that any function which is computable as a subroutine of a \co\Mod[k]\Log\ algorithm belongs to \FUL[k].
For instance, let $\mathbf T$ be an oracle machine implementing a $\Log^{\Mod[k]\Log}$ algorithm for \LCON[k] which simply queries an \LCON[k] oracle and writes the result to the work tape.
Suppose there is a nondeterministic machine $\mathbf T'$ implementing a \co\Mod[k]\Log\ algorithm (with the acceptance conditions of Lemma~\ref{lemma:coMod-k-alt}) which simulates $\mathbf T$.
Whatever decision procedure is performed on the oracle's output, it must be possible for $\mathbf T'$ to accept with one branch (modulo $k$); thus the oracle itself may be simulated in such a way that there are a number of computational branches which coprime to $k$, and without loss of generality equal to $1 \pmod{k}$.
If $\mathbf T'$ copies the oracle result to the output tape after simulating $\mathbf T$ with the appropriate conditions to obtain a single accepting branch modulo $k$, it follows that $\mathbf T'$ is a \FUL[k] machine computing $\chi_k$.
Thus closure of \Mod[k]\Log\ under oracles would have a catastrophic effect on the classes \Mod[p]\Log\ for $p$ dividing $k$.
Let $p_1, p_2, \ldots, p_\ell$ be the distinct prime factors of $k$: from Theorem~\ref{thm:FULnormalForm} and Lemma~\ref{lemma:FULinFun}, we have
\begin{equation}
  \FUL[k] = \Fun\Mod[p_1\!]\Log \;\inter\; \Fun\Mod[p_2\!]\Log \;\inter\; \cdots \;\inter\; \Fun\Mod[p_\ell]\Log\;.
\end{equation}
From $\chi_k \in \FUL[k]$, it follows that $\LCON[k] \in \Mod[p_j]\Log$ for each $p_j$\,, and that therefore $\Mod[k]\Log = \Mod[p_j]\Log$ for every prime factor $p_j$ of $k$.
This does not seem likely for $k$ divisible by multiple primes, unless these classes are also identical to some other class, such as \UL.
We might therefore expect the classes \Mod[k]\Log\ \emph{not} to be closed under oracles, on this basis.
Of course, a proof that $\Mod[p]\Log \ne \Mod[q]\Log$ for primes $p \ne q$ would imply that $\Log \ne \P$, so a proof that \Mod[k]\Log\ is not closed under oracles (for $k$ having multiple prime factors) should perhaps be expected to be difficult.

Note that the result $\LCON[k] \in \co\Mod[k]\Log$ is equivalent to the partial function $\varsigma_k = \LCON[k] \x \{1\}$ being in $\Fundot\co\Mod[k]\Log$.
It is not difficult to show that $\chi_k \in \Fundot\co\Mod[k]\Log$ is equivalent to $\Mod[k]\Log$ being closed under complementation, as follows.
If $\chi_k \in \Fundot\co\Mod[k]\Log$, we could by that fact verify instances of $\overline{\LCON[k]}$ by $\co\Mod[k]\Log$ algorithms; conversely, having \co\Mod[k]\Log\ algorithms for both \LCON[k] and its complement allows us also to verify values of $\chi_k$.
Szelepcs\'enyi~\cite{Sz99} attempts to show that \Mod[k]\Log\ is closed under complementation if and only if it is closed under oracles; however, his approach seems to rely on there being a complete language $L \in \co\Mod[k]\Log$ whose characteristic function $\chi_k$ is in \num[k]\Log\ (which is in fact equivalent to \Mod[k]\Log\ being closed under oracles).
We discuss this and Ref.~\cite{Sz99} in the Appendix.

We conclude with two questions:
\textbf{(a)}~Does $\FUL[k] = \Fun\Mod[k]\Log \;\inter\; \Fundot\co\Mod[k]\Log$ for each $k \ge 2$, and if not, how can we characterize \FUL[k] as a subset?
\textbf{(b)}~Does \FUL\ (the class of functions evaluatable by a \UL\ machine) equal the intersection of $\FUL[k]$ over all $k \ge 2$?

\subsection*{Acknowledgements}

I would like to thank Bjarki Holm for feedback in the early stages of work on this problem, and for indicating helpful references in the literature on the variable modulus problem \LCON; and for the helpful remarks of various anonymous reviewers.

\appendix
\bgroup
\renewcommand\thesection{Appendix.}
\section{Remarks on a preprint of Szelepcs\'enyi}
\egroup
\label{apx:Szelepscenyi}

In an apparently unpublished draft~\cite{Sz99}, Szelepcs\'enyi demonstrated that for arbitrary $k \ge 2$, the hierarchy of classes
\begin{equation}
		\Mod[k]\mathsf{LH}	\;:=\;	\Mod[k]\Log \;\union\; \Mod[k]\Log^{\Mod[k]\Log} \;\union\; \Mod[k]\Log^{\Mod[k]\Log^{\Mod[k]\Log}} \;\union\; \cdots
\end{equation}
is equal to $\AC^0(\Mod[k]\Log)$, the closure of \Mod[k]\Log\ under $\AC^0$ reductions.
He attempted to further demonstrate that $\Mod[k]\mathsf{LH} = \Log^{\Mod[k]\Log}$, and that this implied that closure of \Mod[k]\Log\ under complementation, oracle calls, or completely general $\NC^1$ reductions were equivalent conditions.
In view of the discussion following Lemma~\ref{lemma:coMod-k-alt}, this is a significant claim, which we now discuss.

Szelepcs\'enyi's main focus seems to be on the analogy between \co\Mod[k]\Log\ (which he describes as $\Mod[=k]\Log$) and \Ceq\Log; the former being the modular equivalent of the latter essentially through the relationship described by Lemma~\ref{lemma:numkLogInFun}, using the fact that $\Gap[k]\Log = \num[k]\Log$ by the remarks following Definition~\ref{def:numkLog}.
Using this same equality, wherever Ref.~\cite{Sz99} considers ``acceptance gaps'', we may simply consider the number of accepting branches; however, Szelepc\'enyi's analysis is better motivated by the analogy to \Ceq\Log\ which is suggested by the alternative definition of \co\Mod[k]\Log\ via \Gap\Log\ functions.
Despite the fact that \Ceq\Log\ is also not known to be closed under complementation, Allender, Beals, and Ogihara~\cite{ABO99} show that a similar hierarchy of \Ceq\Log\ algorithms using nested \Ceq\Log\ oracles collapses to $\Log^{\Ceq\Log}$.
Ref.~\cite{Sz99} attempts to show that a similar analysis could be applied to \co\Mod[k]\Log.

We may sketch the main claim of Ref.~\cite{Sz99} as follows (see the original manuscript for details).
Following the proof structure of Ref.~\cite{ABO99} and making implicit use of the characterization of \co\Mod[k]\Log\ in Lemma~\ref{lemma:coMod-k-alt} (albeit expressed in terms of acceptance gaps), Ref.~\cite{Sz99} attempts to show that an $\NC^1(\co\Mod[k]\Log)$ circuit --- involving \textsc{and}, \textsc{or}, \textsc{not}, and $L$ gates for any $L \in \co\Mod[k]\Log$ --- could be simulated in $\Log^{\Mod[k]\Log}$ using the closure of \co\Mod[k]\Log\ under intersection and conjunctive truth-table reductions.
For a given circuit $C$ in a uniform circuit family $\{C_1,C_2,\ldots\} \in \NC^1(\co\Mod[k]\Log)$, one guesses nondeterministically at a set of gates which would produce the outcome `1', where each guess for a gate $g \in C$ is represented by a bit $x_g \in \{0,1\}$ produced during a depth-first traversal of the circuit.
The entire sequence of guesses is represented by a string $x \in \{0,1\}^{|C|}$, and is attributed a mass $m_x$ corresponding to the sum of $2^{d_g} x_g$ over all $g \in C$ (where $d_g$ is the depth of $g \in C$ measured from the output).
In the traversal, some of the gates are simulated by running a \co\Mod[k]\Log\ algorithm as characterized by Lemma~\ref{lemma:coMod-k-alt}, giving rise to some non-trivial number of accepting paths for the entire computation.
\begin{itemize}
\item
	If a gate in $C$ is guessed to have the value `1', it is simulated, essentially multiplying the number of accepting branches modulo $k$ by a \num[k]\Log\ function. 
	If this guess is incorrect (we say a ``false positive''), the number of accepting branches is multiplied by a number which has one or more prime factors in common with $k$; if the guess is correct the number of accepting branches modulo $k$ is unaffected.
\item
	If a gate in $C$ is guessed to have the value `0', it is skipped over in a traversal of the circuit, giving rise to no increase in the number of accepting branches.
	If this guess is incorrect (we say a ``false negative''), the mass $m_x$ is smaller than is might be.
\end{itemize}
One then sets up an optimization problem to try to find the string $x$ with the largest mass, subject to having no false positives, by preventing guesses $x$ with false positives from contributing to the number of accepting paths modulo $k$.
This optimization problem would then have a unique optimum $\bar x$ with mass $\bar m$, which consists of the correct guesses for every gate in $C$, and would be verifiable on a $\Log^{\Mod[k]\Log}$ machine.
This amounts to using a nondeterministic logspace machine $\mathbf U$ to compute functions, where ``incorrect answers'' may not have zero branches modulo $k$ (as with a \FUL[k] machine), but instead may have a number of branches which have prime divisors in common with $k$; and then using this machine as an oracle in a \co\Mod[k]\Log\ algorithm as characterized in Lemma~\ref{lemma:coMod-k-alt}.
Given the unique optimum, one could simply search for it in logarithmic space, allowing a $\Log^{\Mod[k]\Log}$ algorithm to simulate the $\NC^1(\co\Mod[k]\Log)$ machine.
If \Mod[k]\Log\ were then also closed under complements (in which case it would also be closed under arbitrary logspace truth-table reductions), the search for the optimum could be solved by a \Mod[k]\Log\ algorithm, showing that \Mod[k]\Log\ is closed under complements if and only if it is closed under oracles.

The $\Log^{\Mod[k]\Log}$ simulation of $\NC^1(\Mod[k]\Log)$ presented by Ref.~\cite{Sz99} seems to have a flaw, in that it is not clear that that several nondeterministic guesses which each contain false positives could not contribute to simulate the existence of an optimum with mass greater than $\bar m$ for the given verifying algorithm.
While several guesses $x\sur{1}, x\sur{2}, \ldots, x\sur{r}$ at the values of the gates may each have a number of accepting paths $\varphi\sur1, \varphi\sur2, \ldots, \varphi\sur{r}$ which have prime factors in common with $k$, the total number of accepting paths $\varphi\sur1 + \varphi\sur2 + \cdots + \varphi\sur{r}$ may be coprime to $k$, if no integer $\nu > 1$ divides all of the integers $\varphi\sur{j}$\,.

This construction could be repaired if $\chi_k \in \num[k]\Log$, where $\chi_k$ is the characteristic function of a \co\Mod[k]\Log-complete problem $L$, by simulating the nondeterministic machine $\mathbf U$ whose acceptance function is congruent to $\chi_k$ to simulate gates in a $\NC^1(\LCON[k])$ circuit.
Indeed, such a function --- or equivalently, a function $\tilde{\chi}_k$ such that $\tilde\chi_k(x) \equiv 0 \pmod{k}$ when $x \notin L$ and $\gcd(\tilde\chi_k(x),k) = 1$ otherwise --- would make the analogy to the collapse result of Ref.~\cite{ABO99} complete, as the former result depends on the fact that $\Z$ has no zero divisors, and such a function would allow the same simulation technique to avoid all of the zero divisors of $\Z/k\Z$ in the multiplication of branches in the simulation of \co\Mod[k]\Log\ oracles.
However, it is easy to show that $\chi_k \in \num[k]\Log$ if and only if $\chi_k \in \FUL[k]$.
If $\chi_k \in \FUL[k]$, we may simply compute $\chi_k$ and reject if we see that the outcome of the computation is $0$: thus $\chi_k \in \num[k]\Log$.
Conversely, if $\chi_k \in \num[k]\Log$, this means that the partial function $\varsigma_k: \LCON[k] \to \{1\}$ is in $\FUL[k]$; from this we may easily construct another \FUL[k] machine computing $\overline\varsigma_k: \overline{\LCON[k]} \to \{0\}$, and obtain a \FUL[k] machine computing $\chi_k$ by simulating machines computing $\varsigma_k$ and $\overline{\varsigma}_k$ in parallel.
As we saw in Section~\ref{sec:remarks}, $\chi_k \in \FUL[k]$ is already equivalent to \Mod[k]\Log\ being closed under oracles,
which has catastrophic consequences for the mod-logspace classes.

Thus, one might expect the main result of Ref.~\cite{Sz99} to be difficult to salvage, unless one could show that constructive interference of non-deterministic guesses which have false positives could be identified and discounted in the optimisation problem.
In particular, without a correct simulation of the $\NC^1(\Mod[k]\Log)$ circuit corresponding to the unique optimum of that optimisation problem, it is not clear that closure of \Mod[k]\Log\ under complementation is equivalent to closure under oracles.


\bibliographystyle{tocplain}
\bibliography{lcon-mod-k-cjtcs}

\end{document}